\documentclass[11pt]{article}

\usepackage[yyyymmdd,hhmmss]{datetime}
\usepackage{amsmath,amssymb,amsfonts,amsthm}
\usepackage{authblk}
\usepackage{stackengine}
\usepackage{verbatim}
\usepackage{enumerate}

\usepackage{url}
\PassOptionsToPackage{hyphens}{url}
\PassOptionsToPackage{hyphens}{hyperref}
\makeatletter
\g@addto@macro{\UrlBreaks}{\UrlOrds}
\makeatother

\usepackage{graphicx}

\usepackage[dvipsnames]{xcolor}
\allowdisplaybreaks

\usepackage[toc]{glossaries}

\usepackage{tocbibind}

\usepackage{lineno}

\usepackage{algorithm,algorithmicx,algpseudocode}



\floatname{algorithm}{Procedure}

\renewcommand{\algorithmicreturn}[1]{\bgroup\\  ~#1\egroup}
\renewcommand{\algorithmiccomment}[1]{\bgroup\hfill//~#1\egroup}

\theoremstyle{plain} \newtheorem{lemma}{\textbf{Lemma}}
\theoremstyle{plain} \newtheorem{proposition}{\textbf{Proposition}}
\theoremstyle{remark} \newtheorem{remark}{\textbf{Remark}}
\theoremstyle{plain} \newtheorem{theorem}{\textbf{Theorem}}
\theoremstyle{plain} 
\theoremstyle{plain} 
\theoremstyle{plain} 
\theoremstyle{definition} 
\theoremstyle{definition}

\makeatletter
\newcommand{\pushright}[1]{\ifmeasuring@#1\else\omit\hfill$\displaystyle#1$\fi\ignorespaces}
\newcommand{\pushleft}[1]{\ifmeasuring@#1\else\omit$\displaystyle#1$\hfill\fi\ignorespaces}
\makeatother

\makeatletter
\let\@@pmod\pmod
\DeclareRobustCommand{\pmod}{\@ifstar\@pmods\@@pmod}
\def\@pmods#1{\mkern4mu({\operator@font mod}\mkern 6mu#1)}
\makeatother

\newcounter{NbCogito} 

\newcounter{NbFactum} 

\newcounter{NbTabulare} 

\newcounter{NbCoco} 

\newcommand{\rmv}[1]{}

\def\cL{{\mathcal L}}
\def\cN{{\mathcal N}}
\def\cR{{\mathcal R}}

\def\fb{{\mathfrak b}}
\def\fa{{\mathfrak a}}
\def\fc{{\mathfrak c}}

\def\FF{{\mathbb F}}
\def\KK{{\mathbb K}}
\def\LL{{\mathbb L}}
\def\ZZ{{\mathbb Z}}

\def\rT{{\mathrm T}}
\def\sT{{\mathsf T}}

\def\bk{{\mathbf k}}

\providecommand{\myproofname}{Proof}

\title{The complexity of elliptic normal bases} 

\author{Daniel Panario}
\affil
{\stackunder{\small{Carleton University, Canada}}
{\mbox{\small{\texttt{daniel@math.carleton.ca}}}}
}

\author{Mohamadou Sall}
\affil
{\stackunder{\small{Universit\'e Cheikh Anta Diop of Dakar, Senegal}}
{\mbox{\small{\texttt{msallt12@gmail.com}}}}}

\author{Qiang Wang}
\affil
{\stackunder{\small{Carleton University, Canada}}
{\mbox{\small{\texttt{wang@math.carleton.ca}}}}}

\date{}

\begin{document}

\maketitle
\thispagestyle{empty}

\vspace{-1cm}

\begin{abstract}
	We study the complexity (that is, the weight of the multiplication
	table) of the elliptic normal bases introduced by Couveignes and 
	Lercier. We give an upper bound on the complexity of these elliptic
	normal bases, and we analyze the weight of some special vectors related
	to the multiplication table of those bases. This analysis leads us to
	some perspectives on the search for low complexity normal bases from
	elliptic periods.
\end{abstract}

\section{Introduction}

Finite field arithmetic is at the heart of modern information and
communications technology. Indeed the latest generations of Intel 
processors support a hardware instruction for multiplication in 
$\FF_{2^n}$, the finite field of $2^n$ elements \cite{GK08}. The 
main reason behind modern general purpose processors to support 
such arithmetic operations is that finite fields these days appear 
almost everywhere in information processing and telecommunication 
engineering, particularly in error correcting codes and cryptography.

When performing finite field arithmetic, the cost of various operations 
(e.g., addition, multiplication, exponentiation) depends on the choice 
of representation for the elements. For small fields, look-up tables 
can be employed and for mid-sized fields Zech logarithms can be used; 
see \cite[Chapter 10, Table B]{LN86} for an example. However, for 
larger finite fields used in modern public key cryptography, these 
methods are no longer feasible. Thus, for a finite field extension 
$\FF_{q^n}/\FF_q$, it is usual to exploit an explicit isomorphism 
between the field $\FF_{q^n}$ and the corresponding vector space 
$\FF_{q}^n$ by means of a suitably chosen basis. In this paper, we 
are interested in normal bases.

Let $\KK$ be a degree $n$ cyclic extension of a field $\bk$, and 
$\sigma$ a generator of the Galois group $\mathrm{Gal}(\KK/\bk)$. 
Then, a normal basis of $\KK$ over $\bk$ is a basis 
$(\alpha,\sigma(\alpha),\ldots,\sigma^{n-1}(\alpha))$
generated by some $\alpha$ in $\KK\!^*$, the units of $\KK$. Such 
an $\alpha$ is a normal element and the normal basis theorem 
\cite[Theorem 1.4.1]{Gao-Thesis} ensures its existence for all 
finite Galois extension $\KK/\bk$.

Normal bases of $\FF_{q^n}/\FF_q$ have a special property that 
exponentiation by $q$ corresponds to a cyclic shift, and it can be 
computed in negligible time. This is a huge computational advantage 
but multiplication in this context is much more difficult. Let 
$(x_i)_{0\leq i\leq n-1}$ and $(y_i)_{0\leq i\leq n-1}$ be the 
representation of $x$ and $y$ in a normal basis $\cN$ of 
$\FF_{q^n}/\FF_q$. Let $(z_i)_{0\leq i\leq n-1}$ be the coordinates 
of the product $xy$ in $\cN$. Because of the structure of $\cN$, 
the multiplication of basis elements $\alpha^{q^i} \cdot \alpha^{q^j}$ 
are cyclic shifts of each other. It is enough to consider the 
products $\alpha^{q^0} \cdot \alpha^{q^0}, \ldots, \alpha^{q^0} 
\cdot \alpha^{q^{n-1}}$ written in the normal basis.
This defines the matrix multiplication of the normal basis.  
The total number of nonzero elements in this matrix is the 
complexity $C(\cN)$ of $\cN$. The lower the complexity is, the more 
interesting these bases become. Indeed, in practice, the number of 
fanout of cells decreases with this complexity. In hardware 
implementations, normal bases are often preferred; see 
\cite[Section 16.7]{MP13} for more details.

There are different ways of doing efficient finite field arithmetic
using normal bases. For example, the bases described in 
\cite{CL09, ES19} give the product $xy$ in a normal basis from the 
coordinates of $x$ and $y$ without using the multiplication table 
of this basis. One can also extend appropriate normal bases or use 
trace computation to take advantage of their structures. We shall 
not go into details into those studies, so we just refer to 
\cite{CGPT08, ES20, TW18} for these approaches.

A well-known method to construct efficient normal bases is using 
Gauss periods; see \cite{Ash-Blake-Vanstone, Gao-Gathen-Panario-Shoup}.
However, as shown by Wassermann \cite{W90}, Gaussian normal bases 
do not always exist. Couveignes and Lercier \cite{CL09} show how 
these methods can be generalized using elliptic curve groups. 
Since there are many elliptic curves, they enlarge significantly 
the number of cases where a normal basis with fast multiplication 
exists. However, the number of nonzero elements of the multiplication 
table, from their construction is not well understood. In this paper 
we proceed to the analysis of the complexity of normal bases from 
elliptic periods as designed in \cite{CL09}. We give some bounds on 
this complexity using special vectors defined from elliptic curves.

In Section \ref{sec:BM}, we review some background material on 
elliptic curves and normal elements. Our results on the complexity 
of elliptic normal bases are presented in Section \ref{sec:CEP}. 
In Subsection \ref{sec:BCEP}, we prove the main theorem that gives 
upper and lower bounds on this complexity. In Subsection \ref{sec:CSV}, 
we analyze the weight of some special vectors related to the 
multiplication tensor of elliptic normal bases. In Section 
\ref{sec:EENBT}, we support our studies with practical examples 
using the computational algebra system Magma \cite{HMF}. For these 
experiments, we develop a Magma package that defines all the 
parameters of an elliptic normal basis, and then gives bounds on 
its complexity. This Magma package, with test and concrete examples, 
is available online \cite{MagmaPack}. We conclude, in Section 
\ref{sec:KNRP}, with some problems for further research.

\section{Background material}\label{sec:BM}
In this section, we briefly recall definitions and properties of 
elliptic curves and normal bases over finite fields as required 
in this paper.

\subsection{Elliptic curves}
Elliptic curves are originally from algebraic geometry which deals 
with the solutions of polynomial equations. Their subject encompasses 
a vast amount of mathematics. Our aim in this section is to summarize 
just enough of the basic theory as needed for the normal elements 
construction in \cite{CL09}.

An elliptic curve $E$ defined over a field $\KK$ is given by a 
Weierstrass equation of the form 
$$E : y^2 + a_1xy + a_3y = x^3 + a_2x^2 + a_4x + a_6,$$
with $a_1, a_2, a_3, a_4, a_6 \in \KK$. It can be shown, that 
there is a group associated with the points on this curve; see 
\cite{Silv09}, Chapter III, Section 2, for the proof. For any 
field extension $\LL/\KK$ the group of $\LL$-rational points of 
$E$ is the set 
$$E(\LL) = \{(x, y)\in \LL^2 \ : \ E(x, y)=0 \}\cup \{ O\},$$ 
endowed with the usual group law, where $O$ is the identity 
element. An important concept related to elliptic curves is finite 
dimensional linear spaces given by a special divisor on this curve. 
A divisor $D$ on an  elliptic curve $E$ is a formal sum 
$$D=\sum_{P\in E}n_P(P),$$ 
where $n_P\in\ZZ$ and $n_P = 0$ for all but finitely many $P\in E$. 
For two divisors $D_1$ and $D_2$ over $E$, we have that $D_1\geq D_2$ 
if $n_P\geq 0$ for every $P$ in the formal sum given by $D_1-D_2$. 
To every divisor $D$, we can associate a set of functions
$$\cL(D) = \{f\in \bar{\KK}(E)^* : \mathrm{div}(f)\geq-D\}\cup \{ O\},$$
where $\bar{\KK}$ is the algebraic closure of ${\KK}$. Here, the 
formal sum $\mathrm{div}(f)$ is the divisor associated to the 
function $f$, that is,
$$ \mathrm{div}(f) = \sum_{P\in E(\bar{\KK})}\nu_P(f)(P),$$
where $\nu(f)$ is the order of the vanishing of $f$ at the point $P$. 
We have that $\cL(D)$ is a finite-dimensional $\bar{\KK}$-vector 
space, and its importance lies in that 
it gives information about functions on $E$ having prescribed zeros 
and poles. Couveignes and Lercier \cite{CL09} made use of relations 
between elliptic functions in these linear spaces which result in 
nice multiplication formulas for their elliptic basis. 
However, the number of nonzero elements in the multiplication tables
(the complexity) of these elliptic normal bases was not analyzed.

\subsection{Normal bases over finite fields}
Let $\FF_{q^n}$ be a field extension of degree $n$ over $\FF_q$,
and $\FF_{q}^{n}$ an $\FF_q$-vector space of dimension $n$.
Let $\alpha_0, \alpha_1, \ldots, \alpha_{n-1} \in \FF_{q^n}$
be a basis of this space. Then for all $A \in \FF_{q^n} $
$$A=\sum_{i=0}^{n-1} a_i\alpha_i,$$
where $a_i \in \FF_q$. We write $A=(a_0, a_1, \ldots, a_{n-1})$.
Let $B=(b_0, b_1, \ldots, b_{n-1})$ be another element of 
$\FF_{q^n}$ and $C=(c_0, c_1,\ldots, c_{n-1})$ the product
\begin{align*}
	A\cdot B    = \sum_{i,j} a_i b_j\alpha_i\alpha_j.
\end{align*}
Since $\alpha_i\alpha_j\in\FF_{q^n}$, one can set, for $t_{ij}^{(k)}\in\FF_q$, 
$$ \alpha_i\alpha_j = \sum_{k=0}^{n-1} t_{ij}^{(k)}\alpha_k.$$
Let $ \rT_k = (t_{ij}^{(k)}) $ be an $n \times n $ matrix over 
$\FF_q$ independent of $A$ and $B$, and let $B^t$ be the transpose 
of $B$. We have that
$$ c_k = \sum_{0\leq i,j < n} a_i b_jt_{ij}^{(k)} = A\rT_kB^t, \ \ 0 \leq k < n.$$

A normal basis $\mathcal{N}$ of $\FF_{q^n}$ over $\FF_q$ is a 
basis of the form $\{\alpha, \alpha^q,\ldots, \alpha^{q^{n-1}} \}$. 
In this case, from one matrix $\rT_0$ one can get the properties 
of the remaining matrices in the collection $(\rT_k)_{0\leq k\leq n-1}$. 
This is a main advantage of a normal basis. The total number of 
nonzero structural constants in $\rT_0$ is very important and 
should be kept as small as possible for efficient software and 
hardware implementation; see \cite{Ash-Blake-Vanstone,Gao-Thesis}. 
The number of nonzero elements of $\rT_0$ is the complexity of $\cN$.

\begin{proposition}[\cite{MOVW89}]\label{prop:opt:nbasis}
	The complexity $C_{\cN}$ of a normal basis $\cN$ of $\FF_{q^n}$ over 
	$\FF_q$ is bounded by $$2n-1\leq C_{\cN}\leq n^2-n+1.$$
\end{proposition}
A normal basis is optimal if it achieves the lower bound in Proposition 
\ref{prop:opt:nbasis}, and a normal basis has a low complexity if it 
has sub-quadratic bounds, with respect to $n$, on their complexity.
It turns out that most normal bases do not have low complexity 
\cite{MMPT08, MPT18}, and finding normal elements with low complexity 
is an interesting research problem \cite[Problem 6.2]{Gao-Thesis}.
For example, the efficiency of multiplication of some (non-normal) 
bases depends on the complexity of a well-chosen normal element;
for more details, see \cite{ES20,TW18}.

\subsection{Elliptic normal bases}

There are different ways of constructing normal elements. Couveignes 
and Lercier \cite{CL09} were the first to study normal bases from 
elliptic periods. Their construction is derived from elliptic curves 
properties. Let $v_l$ be the valuation associated to $l$, that is, 
$v_l(n)$ is the degree of divisibility of $n$ by $l$. The existence 
of an appropriate elliptic curve depends on a unique positive integer 
$n_q$ satisfying for every prime number $l$
\begin{equation}
	\begin{cases}
		v_l(n_q)=v_l(n) \text{ if } l \text{ is coprime with } q-1,\\
		v_l(n_q)=0 \text{ if } v_l(n)=0, \\
		v_l(n_q)=\mathrm{max}\{2v_l(q-1)+1, 2v_l(n)\} \text{ if } l|q-1 \text{ and } l|n.\\
	\end{cases}
\end{equation}
The following theorem gives a brief summary of their results; 
for more details, see \cite[Section 5]{CL09}.

\begin{theorem}[\cite{CL09}]\label{the:main:theo:couv-lerc}
	With the above notation, let $p$ be a prime number and $q$ a power 
	of $p$. Let $n\geq2$ be an integer such that $$n_q\leq \sqrt{q}.$$
	Then, there exists an elliptic curve $E$ over $\FF_q,$ a point
	$t\in E(\FF_q)$ of order $n$ and a point $b\in E(\bar{\FF}_q)$ such 
	that $\phi(b)=b+t$ and $nb\neq 0$, where $\phi$ is the Frobenius
	automorphism. 
	
	Furthermore, there exists a point $R\in E(\FF_q)$ such that $nR\neq O$, 
	and a normal basis $\cN$ of $\FF_{q^n}/\FF_q$ with quasi-linear cost 
	of multiplication. 
\end{theorem}

\begin{remark}
	We keep the term ``complexity'' to refer to the number of non-zero elements
	in the multiplication table of a normal basis and use ``cost" for the running
	time of algorithms. We observe that in 
	\cite{CL09} the cost of the algorithm is called the complexity of 
	the operation. For vector, we use weight to represent the number 
	of nonzero elements of the vector.
\end{remark}

The methods of Couveignes and Lercier allow a fast multiplication 
algorithm but the complexity of their normal bases, especially 
needed for hardware implementations, was not analyzed. The main goal 
of this paper is to study the complexity of elliptic normal bases.

\section{The complexity of elliptic normal bases}\label{sec:CEP}

In this section, we give the main results of this paper. It may be 
convenient for the reader to have in hand Couveignes and Lercier's 
paper \cite{CL09}. Section 4, where the multiplication in the 
elliptic normal basis is given plays a central role in our paper.

\subsection{Notation and previous results}

We use the following notation. For $P$ a point on an elliptic 
curve $E$, we denote the $x$-coordinate of $P$ by $x_O(P)$ or 
simply $x(P)$. For $x$ and $y$, two elements of $\FF_{q^n}$, 
we denote by $\overrightarrow{x}$ and $\overrightarrow{y}$ their 
coordinates in the elliptic normal basis. We denote by 
$\overrightarrow{x}\star \overrightarrow{y}$ the convolution 
product and by 
$\overrightarrow{x} \diamond \overrightarrow{y}=(x_k y_k)_{k}$
the component-wise product of $\overrightarrow{x}$ and 
$\overrightarrow{y}$. We denote by 
$\sigma(\overrightarrow{x})=(x_{k-1})_{k}$, with $(x_{n-1})_{0}$,
the cyclic shift of $\overrightarrow{x}$. For a vector 
$\overrightarrow{\cR}$ of length $n$, we denote by 
$\overleftarrow{\cR}_k$ the vector given by the component-wise 
product between $\overrightarrow{\cR}$ and its $k$-cyclic shift; 
we remark the use of left-to-right arrow for this operation. 
It is clear that the complexity of $\overleftarrow{\cR}_k$ is 
less than or equal to the complexity of $\overrightarrow{\cR}$. 

To differentiate between the complexity $C_{\cN}$ of the normal 
basis and the weight of a vector, we denote by $w(\overrightarrow{\cR})$ 
the weight, that is, the number of nonzero elements of a vector 
$\overrightarrow{\cR}$.

We require the following result in our analysis.

\begin{lemma}{\cite[Lemma 6]{CL09}} \label{mainlemma}
	The multiplication tensor for the normal elliptic basis is
\begin{eqnarray*}
	&& \hspace{-0.25cm} 
	(\overrightarrow{\iota} \mathfrak{a}^2)\star((\overrightarrow{x}
	-\sigma(\overrightarrow{x}))\diamond (\overrightarrow{y}
	-\sigma(\overrightarrow{y})) \\
	&+& \hspace{-0.25cm} 
	\overrightarrow{\cR}^{-1}\star((\overrightarrow{\cR}\star 
	\overrightarrow{x})\diamond (\overrightarrow{\cR}\star \overrightarrow{y})
	-(\overrightarrow{\cR}_x \mathfrak{a}^2)\star
	((\overrightarrow{x}-\sigma(\overrightarrow{x}))\diamond
	(\overrightarrow{y}-\sigma(\overrightarrow{y}))))),
\end{eqnarray*}
where $\overrightarrow{\cR}=(\fa f_{O, t}(R+jt)+\fb)_{0\leq j\leq n-1}$,
$\overrightarrow{\cR}_x=(x_O(R+jt))_{0\leq j\leq n-1}$, and $\overrightarrow{\iota}$
represents the coordinates of $x(b)$ in $\cN$. The elements $\fa$ and $\fb$ are two scalars in $\FF_q$. 

In particular, it consists of 5 convolution products, 2 component-wise 
products, 1 addition, and 3 subtractions between vectors of size $n$ in 
$\FF_q$.
\end{lemma}

The scalars $\fa$ and $\fb$ have the special property $$\fa\fc+n\fb=1,$$ where $\fc$ is also a constant defined by the 
sum of the elliptic functions $(f_{kt, (k+1)t})_{0\leq k\leq n-1}$. It is shown in \cite[Lemma 4]{CL09} 
that the sum of these functions is a constant.

\begin{remark}
We observe that some notation is different from \cite{CL09} and that 
in the previous lemma, we combined some of their results for simplicity 
of presentation; see \cite[Section 4]{CL09} for more details. In 
particular, the vector $\overrightarrow{\cR}$, that is central to the
complexity analysis is the evaluation at certain special points of 
an elliptic function $f$ to be defined in 
Theorem~\ref{The:Complexity:EllPeriod}.
\end{remark}

\subsection{Bounds on the complexity of elliptic normal bases}\label{sec:BCEP}

We show that the complexity of elliptic periods defined in Theorem 
\ref{the:main:theo:couv-lerc} depends (in some precise way) on a 
vector $\overrightarrow{\cR}$ derived from two $\FF_q$-rational 
points $t$ and $R$ over an elliptic curve $E$. The point $t$ is 
of order $n$ and $R$ lies outside of $E[n]$, the subgroup of 
$n$-torsion points of $E$ 
$$E[n] = \{ P \in E(\FF_q) \colon [n]P=0\}.$$ 

The following theorem gives a bound of the whole complexity of 
the elliptic normal basis in term of these two vectors.

\begin{theorem}\label{The:Complexity:EllPeriod}
With the same notation as before, let $p$ be a prime number and 
$q$ a power of $p$. Let $n\geq2$ be an integer such that 
$n_q\leq \sqrt{q}$. Then, there exists an elliptic curve $E$ 
over $\FF_q,$ a $\FF_q$-rational point $t$ of order $n$ and a 
point $R\in E(\FF_q)$ lying outside of $E[n]$.

Furthermore, there exists a normal basis 
${\cN}=\{\alpha_0, \alpha_1, \ldots, \alpha_{n-1}\}$ of 
$\FF_{q^n}/\FF_q$ with complexity lower and upper bounded by 
$$ 3+ \sum_{k=2}^{n-2}w(\overrightarrow{\cR}^{-1}\star\overleftarrow{\cR}_k)
\leq C({\cN}) \leq
3n + \sum_{k=2}^{n-2}w(\overrightarrow{\cR}^{-1}\star\overleftarrow{\cR}_k),$$
where the terms in the sum depend on $R$ and $t$.
\end{theorem}

\begin{proof}
The existence of the elliptic curve, the first part of the theorem, 
follows from Theorem~\ref{the:main:theo:couv-lerc}. We focus 
next on the nonzero entries of the multiplication table of the 
elliptic normal bases.

Let $A$ and $B$ be two distinct points on the elliptic curve $E$
and consider $\bar{\FF}_q$, the algebraic closure of $\FF_q$.
There exists an elliptic function $f_{A, B}$ having two poles 
at $A$ and $B$ given by  
\[ \begin{array}{llll}
	f_{A, B} : & E(\bar{\FF}_q) & \longrightarrow & \bar{\FF}_q\\
	& P & \longmapsto & \dfrac{y_A(P)-y(A-B)}{x_A(P)-x(A-B)}
\end{array} \]
where $y_A(P)$ and $x_A(P)$ represent the $y$-coordinate and the 
$x$-coordinate of the point $P-A$, respectively. The evaluation 
of the function $f_{A, B}$ at $P$ is merely the slope of the line 
passing through $P-A$ and $A-B$. Since $n_q\leq\sqrt{q}$, it is 
known from Lemma~\ref{mainlemma} 
that the multiplication tensor for the normal elliptic basis 
obtained in Theorem \ref{the:main:theo:couv-lerc} consists of 
$5$ convolution products, $2$ component-wise products, $1$ 
addition and $3$ subtractions between vectors of size $n$ in $\FF_q$. 

Indeed, and more precisely, for two vectors $\overrightarrow{x}$ 
and $\overrightarrow{y}$ representing two elements $x$ and $y$ to 
be multiplied in the elliptic normal basis $\cN$, the coordinates 
of the product $x\times y$ are given by 
\begin{eqnarray*}
&& \hspace{-0.25cm} 
(\overrightarrow{\iota} \mathfrak{a}^2)\star((\overrightarrow{x}
-\sigma(\overrightarrow{x}))\diamond (\overrightarrow{y}
-\sigma(\overrightarrow{y})) \\
&+& \hspace{-0.25cm} 
\overrightarrow{\cR}^{-1}\star((\overrightarrow{\cR}\star 
\overrightarrow{x})\diamond (\overrightarrow{\cR}\star \overrightarrow{y})
-(\overrightarrow{\cR}_x \mathfrak{a}^2)\star
((\overrightarrow{x}-\sigma(\overrightarrow{x}))\diamond
(\overrightarrow{y}-\sigma(\overrightarrow{y}))))),
\end{eqnarray*}
where $\overrightarrow{\cR}=(\fa f_{O, t}(R+jt)+\fb)_{0\leq j\leq n-1}$,
$\overrightarrow{\cR}_x=(x_O(R+jt))_{0\leq j\leq n-1}$, and 
$\overrightarrow{\iota}$ represents the coordinates of $\alpha_0^2$ 
in $\cN$. The elements $\fa$ and $\fb$ are two scalars in $\FF_q$. 
For $0\leq k\leq n-1$, the cross-product $\alpha_0 \alpha_k$ is given by
\begin{eqnarray}\label{eq:enb-mult-tensor}
	&& \hspace{-0.35cm}
	\overrightarrow{\iota}\star \left(\fa^2(\overrightarrow{\alpha_0}-
	\sigma(\overrightarrow{\alpha_0}))\diamond((\overrightarrow{\alpha_k}-
	\sigma(\overrightarrow{\alpha_k}))\right)\\
	&+& \hspace{-0.35cm}
	\overrightarrow{\cR}^{-1}\star\left((\overrightarrow{\cR}\star 
	\overrightarrow{\alpha_0})\diamond 
	(\overrightarrow{\cR}\star \overrightarrow{\alpha_k}) 
	-\overrightarrow{\cR}_x \star\left(\fa^2(\overrightarrow{\alpha_0}-
	\sigma(\overrightarrow{\alpha_0}))\diamond((\overrightarrow{\alpha_k}-
	\sigma(\overrightarrow{\alpha_0}))\right)\right). \nonumber
\end{eqnarray}

We observe that $\overrightarrow{\alpha}_0=(1, 0, \ldots, 0), \ 
\sigma(\overrightarrow{\alpha}_0)=(0, 1, \ldots, 0), \text{ and } \alpha_k=(0,0,
\ldots,1,0,\ldots,0) $. Hence, we have
\begin{eqnarray*}
\overrightarrow{\alpha}_0-\sigma(\overrightarrow{\alpha}_0) 
&=& (1,-1,0,0,\ldots,0,0),\\
\overrightarrow{\alpha}_1-\sigma(\overrightarrow{\alpha}_1) 
&=& (0,1,-1,0,\ldots,0,0),\\
\overrightarrow{\alpha}_2-\sigma(\overrightarrow{\alpha}_2) 
&=& (0,0,1,-1,0,\ldots,0),\\
&\vdots& \\
\overrightarrow{\alpha}_{n-1}-\sigma(\overrightarrow{\alpha}_{n-1}) 
&=& (-1,0,0,0\ldots,0,1).
\end{eqnarray*}
When $2\leq k\leq n-2$ we have
\begin{eqnarray*}
&& (\overrightarrow{\alpha_0}-\sigma(\overrightarrow{\alpha_0}))\diamond
((\overrightarrow{\alpha_k}-\sigma(\overrightarrow{\alpha_k})) \\
&=&
(1,-1,0,\ldots,0)\diamond(0,0,\ldots,1,-1,0,\ldots,0)=(0,\ldots,0).
\end{eqnarray*}
Since $\overrightarrow{\alpha_0}=(1,0,\ldots,0)$ can be viewed as the
neutral element of the convolution product, the study of the 
complexity of the normal basis is basically reduced to the analysis
of the vector
$$\overrightarrow{\cR}^{-1}\star\left(\overrightarrow{\cR}\diamond 
(\overrightarrow{\cR}\star \overrightarrow{\alpha_k})\right). $$
To reduce the vector 
$\overrightarrow{\cR}^{-1}\star\left(\overrightarrow{\cR}\diamond 
(\overrightarrow{\cR}\star \overrightarrow{\alpha_k})\right)$ we remark
that the convolution product $\overrightarrow{\cR}\star \overrightarrow{\alpha_k}$
is just a $k$-right cyclic shift of the vector $\overrightarrow{\cR}$. 
In other words, if $\overrightarrow{\cR} = (r_0, r_1, \ldots, r_{n-1})$, then
$$\overrightarrow{\cR}\star \overrightarrow{\alpha_k}=(r_{n-k}, r_{n-k+1}, 
r_{n-k+2}, \ldots, r_{2n-k-1}),$$
where indices are computed modulo $n$. Therefore
\begin{eqnarray*}
&& \overrightarrow{\cR}^{-1}\star\left(\overrightarrow{\cR}\diamond 
(\overrightarrow{\cR}\star \overrightarrow{\alpha_k})\right) \\
&=& \overrightarrow{\cR}^{-1}\star ((r_0, r_1, \ldots, r_{n-1})\diamond 
(r_{n-k}, r_{n-k+1}, r_{n-k+2}, \ldots, r_{2n-k-1}))\\
&=& \overrightarrow{\cR}^{-1} \star \overleftarrow{\cR}_k.
\end{eqnarray*}
For the $n-3$ rows in the middle of the multiplication table of $\cN$, 
we get the total complexity 
$$\sum_{k=2}^{n-2}w(\overrightarrow{\cR}^{-1} \star \overleftarrow{\cR}_k).$$
It remains to note that the complexity of the first two rows and the 
last row is lower and upper bounded by $3$ and $3n$, respectively, 
since each of the remaining rows must have at least one 
nonzero element. 
\end{proof}

We observe that from the points $R$ and $t$, we know the exact 
complexity of the matrix derived from the $n-3$ rows at the middle 
of the multiplication table of $\cN$. Next, we study the weight of 
the vectors involved in this complexity estimation.

\subsection{Weight of special vectors}\label{sec:CSV}

We study special vectors, that is, vectors constructed from elliptic 
curve properties and on which the efficiency of the elliptic normal basis 
depends.

In the previous subsection, we showed that the complexity of the elliptic
normal bases, with fast multiplication algorithm, designed by Couveignes 
and Lercier \cite{CL09} depends heavily on the vectors 
$\overrightarrow{\cR}, \ \overrightarrow{\cR}_x, \ \overleftarrow{\cR}_k,
\ \overrightarrow{\cR}^{-1}$ and $\overrightarrow{\iota}$, the
expression of $x(b)$ in the normal basis.

In this section, we give bounds on the weights of the first three vectors 
showing that all of them are very close to $n$. Unfortunately, it is 
difficult to predict the weights 
$w(\overrightarrow{\cR}^{-1}*\overleftarrow{\cR}_{k})_{2\leq k\leq n-2}$
needed in Theorem~\ref{The:Complexity:EllPeriod}. However, we conjecture
that those weights are also very high. We show this with some examples 
in the next section.

The vectors $\overrightarrow{\iota}$
and $\overrightarrow{\cR}_x$ are involved in the computation of $C_{\cN}$ 
when $k=1$ and $n-1$. Indeed, if $k=1$, then 
Equation (\ref{eq:enb-mult-tensor}) becomes
\begin{align}\label{eq:row-one}
\alpha_0\alpha_1 =&\overrightarrow{\iota}\star \left(0,-\fa^2,0,\ldots,0\right)+
\overrightarrow{\cR}^{-1}\star\left(\overleftarrow{\cR}_1 
-\overrightarrow{\cR}_x \star\left(0,-\fa^2,0,\ldots,0\right)\right),
\end{align}
and if $k=n-1$, we have
\begin{align}\label{eq:row-last}
\alpha_0\alpha_{n-1} =&\overrightarrow{\iota}\star \left(-\fa^2,0,\ldots,0\right)+
\overrightarrow{\cR}^{-1}\star\left(\overleftarrow{\cR}_{n-1} 
-\overrightarrow{\cR}_x \star\left(-\fa^2,0,\ldots,0\right)\right)
\end{align}

To have a better understanding of $C_{\cN}$, we look at the
structure of the vectors $\overrightarrow{\cR}, \ \overleftarrow{\cR}_k$ and
$\overrightarrow{\cR}_x$. We need the following proposition that plays a
role in our analysis.

\begin{proposition}[\cite{Silv09}, Proposition II.3.1]
Let $C$ be a smooth curve defined over a field $\KK$ and $f\in\bar{\KK}(C)^*$,
with $\bar{\KK}(C)$ the function field of $C$ over $\bar{\KK}$. Then $\mathrm{deg(div}(f))=0$.
\end{proposition}
The degree of $\mathrm{div}(f)$, denoted by $\mathrm{deg(div}(f))$, 
is the sum of the number of zeroes and poles of $f$. As a consequence, 
this proposition implies that the number of zeroes of a function $f$ 
defined over an elliptic curve $E$ is equal (in terms of points) to 
the number of its poles counted with the opposite of their order of 
multiplicity. We note that if $P$ is a pole of a function $f$, then 
the order $\nu(f)$ of the vanishing of $f$ at $P$ is a negative 
integer. Hence, the poles are counted with a negative sign.

\bigskip

\noindent
{\bf Weights $w(\overrightarrow{\cR})$ and $w(\overleftarrow{\cR}_k)$}:
Let $f_0(R+jt) := \fa f_{O, t}(R+jt)+\fb$. 
Since $\overrightarrow{\cR}=(f_0(R+jt))_{0\leq j\leq n-1}$, then the number 
of nonzero elements in $\overrightarrow{\cR}$ depends on the evaluation
of $f$ at the points $R+jt$.
Recall that $\fa$ and $\fb$ are two scalars in $\FF_q$ such that $\fa\fc+n\fb=1$.
If $\fb=0$, since $\fa$ is a nonzero element, $f_0(R+jt)=0$ if and only
if $f_{O, t}(R+jt)=0$. Since $f_{O, t}$ has only two poles at $O$ and $t$ then 
it has at most two zeroes that might be in the set $(R+jt)_{0\leq j\leq n}$.
Hence, we have that if $\fb=0$ then $w(\overrightarrow{\cR})\geq n-2$. 
Following the same reasoning we obtain that if $\fb\neq0$ then $w(\overrightarrow{\cR})\geq n-1$.

We can easily derive the weight of $\overleftarrow{\cR}_k$ from 
$w(\overrightarrow{\cR})$. By definition $\overleftarrow{\cR}_k$ is the 
component-wise product between $\overrightarrow{\cR}$ and its $k$-cyclic 
shift. Thus, if $\overrightarrow{\cR}$ has one zero element, then 
$w(\overleftarrow{\cR}_k)=n-2$, and if $\overrightarrow{\cR}$ has two zeroes,
then $w(\overleftarrow{\cR}_k)\geq n-4$.

\bigskip

\noindent
{\bf Weight $w(\overrightarrow{\cR}_x)$}: We have that 
$\overrightarrow{\cR}_x=(x_O(R+jt))_{0\leq j\leq n-1}$. This vector is given by the $x$-coordinates 
of $n \ \FF_q$-rational points $R+jt$ over $E$. If $\mathrm{char}(\FF_q)\neq 
2, 3$, then the Weierstrass form of $E$ can be reduced to the following 
equation
\begin{align}
y^2=x^3+ax+b,
\end{align}
with $a, b\in \FF_q$. Every point $P\in E$ on the $y$-axis has $\pm\sqrt{b}$ 
as $y$-coordinate. Since we have two possibilities, the number of zero-elements
of $\overrightarrow{\cR}_x$ can not exceed $2$, in other words 
$w(\overrightarrow{\cR}_x)\geq n-2$.

\bigskip

We summarize the preceding discussion in the following proposition, which gives 
some bounds on the weight of the considered vectors. We suppose the 
conditions of Theorem \ref{The:Complexity:EllPeriod} are satisfied.
\begin{proposition}\label{Lem:Weight:Vectors}
Let $E$ be an elliptic curve over $\FF_q$ giving rise to the two vectors
$\overrightarrow{\cR}$ and $\overrightarrow{\cR}_x$. Then,  
$w(\overrightarrow{\cR}_x)\geq n-2$, $w(\overrightarrow{\cR})\geq n-2$,
and
$$ 
w(\overleftarrow{\cR}_k)\geq
\begin{cases}
\vspace{0.2cm}
n-2,  \text{ if } w(\overrightarrow{\cR}) \text{ has one zero element};\\
n-4,  \text{ if } w(\overrightarrow{\cR}) \text{ has two zeroes.}
\end{cases}
$$
\end{proposition}

The main consequence of this proposition is that most of the vectors 
involved in the multiplication are highly non-sparse. It remains to 
understand the weights of the $n-3$ convolutions
$w(\overrightarrow{\cR}^{-1}*\overleftarrow{\cR}_{k})_{2\leq k\leq n-2}$.

\section{Examples of elliptic normal basis}\label{sec:EENBT}

In order to have a practical understanding of our results in this 
paper, we give three examples of elliptic normal bases using 
different parameters and/or extensions. For the sake of better 
illustration, we work with extensions of small degrees. All these 
examples were built using a Magma package which contains two main 
functions. 
\begin{enumerate}
\item The function {\it ENBparamsComputation} which takes as
input a finite field $\FF_q$ and an extension degree $n$ and 
then returns (if possible) all the necessary parameters to the 
construction of an elliptic normal basis of $\FF_{q^n}$ over $\FF_q$.
\item The function {\it ENBcomplexityBounds} which computes the 
weight of the special vectors and gives a lower and an upper bound 
on the complexity of the basis.

\end{enumerate}
These functions make use of several preliminary functions; for more 
details on this package, see \cite{MagmaPack}.

An interesting point of working with elliptic curves is that for 
a given finite field $\FF_q$ we can define many such curves. 
This can enlarge significantly the number of cases where a normal 
basis with good properties (fast multiplication and/or low complexity) 
exist. We highlight this fact in Subsection \ref{sub:SE} by giving 
two different elliptic normal bases on the same extension 
$\FF_{7^6}$ over $\FF_7$. The special vectors (on which the 
complexity depends heavily) of one of the bases are full of nonzero 
elements, which leads to a bad complexity. In the last example,
the change of an elliptic curve gives some special vectors with 
good weight, which leads to a lower bound on the complexity of 
the constructed basis. 

\subsection{First example}\label{sub:FE}

\indent Let us consider the finite field $\FF_{13}$ and degree 
extension $n=7$. In this subsection, we give a practical example
of an elliptic normal basis, determine the weight of its special 
vectors and give a bound on its complexity. We work with the 
elliptic curve 
$$E : y^2+4xy+9y = x^3+x^2+3x+8,$$ 
of order $14$ over $\FF_{13}$ and the subgroup $\sT$ generated by 
the point $t=(0, 10)$ of order $7$. The quotient isogeny $E/\sT$ 
is defined by the equation 
$$y^2+4xy+9y = x^3+x^2+6,$$ 
and its 
rational point $a=(11, 11)$ generates an irreducible preimage from 
which we construct the extension $\FF_{13^7}$. The generic point
$$b=(\tau^5 + 10\tau^4 + 9\tau^3 + 9\tau^2 + 6\tau + 8, 12\tau^6 + 3\tau^5 + 10\tau^4 + \tau^3 + 10\tau^2 + 11\tau + 1)$$ 
is an $\FF_{13^7}$-rational 
point of $E$. The evaluation of the functions $\fa f_{kt, (k+1)t}+\fb$ 
at this point $b$ gives the elliptic normal basis 
$\cN_1 = \{\alpha_0, \alpha_1, \ldots, \alpha_6 \}$ where 
\begin{align*}
\alpha_0 & = 12\tau^6 + 9\tau^3 + \tau^2 + 12\tau + 5,\\
\alpha_1 & = 12\tau^6+3\tau^3+8\tau^2+7\tau,\\
\alpha_2 & = 3\tau^5+7\tau^4+8\tau^3+8\tau^2+9\tau+5,\\
\alpha_3 & = \tau^6+2\tau^5+9\tau^4+8\tau^3+11\tau^2+3\tau+3,\\
\alpha_4 & = \tau^6+11\tau^5+5\tau^4+9\tau^3+11\tau^2+10\tau+3,\\
\alpha_5 & = \tau^6+\tau^5+3\tau^4+7\tau^3+7\tau^2+7\tau+5,\\
\alpha_6 & = 12\tau^6+9\tau^5+2\tau^4+8\tau^3+6\tau^2+4\tau+11.
\end{align*}
In this example, the scalars $\fc, \fa,$ and $\fb$ in Lemma \ref{mainlemma} are respectively given by $11, 6$ and $0$, and only
$\fa$ is needed in the computations of this section. To compute the special vectors 
and determine the bound on the complexity $C(\cN_1)$ of $\cN_1$ 
we define the $\FF_{13}$-rational point $R=(9, 0)$ over $E$. The point $R$ is outside
of $\sT$. The complexity $C(\cN_1)$ of $\cN_1$ satisfies $$25\leq C(\cN_1)\leq 43.$$
These two bounds are derived from the weight of the special vectors
in Table~\ref{table:specvects1}.
\begin{table}[ht]
\centering 
\vspace{0.3cm}

\begin{tabular}{cc||c} 
\hline 
& & 
($\overrightarrow{\cR}^{-1}*\overleftarrow{\cR}_k)_{2\leq k\leq n-2}$ \\ [0.5ex] 
\hline 
$\overrightarrow{\cR}_x$ & (9, 3, 6, 1, 6, 3, 9) & (3, 5, 3, 11, 11, 11, 11) \\ 
$\overrightarrow{\cR}$   & (4, 0, 8, 10, 10, 8, 0) & (6, 1, 1, 6, 0, 0, 0) \\
$\overrightarrow{\cR}^{-1}$ & (12, 8, 6, 0, 0, 8, 8) & (6, 0, 0, 0, 6, 1, 1) \\
&    &    (3, 11, 11, 11, 11, 3, 5) \\
\hline 
\end{tabular}

\caption{Special vectors of $\cN_1$.}
\label{table:specvects1} 
\end{table}

Using our Magma function, \textit{ENB3RowsWeight}, that we defined for
these toy examples, we can compute the coordinates of $\alpha_0^2$ in 
$\cN_1$, the vector $\overrightarrow\iota$, and by extension the exact 
complexity of $\cN_1$. We have that 
$(\overrightarrow{\alpha_0^2})=(9, 6, 3, 11, 1, 5, 10)$ and 
$\overrightarrow{\iota}=(10, 10, 9, 5, 1, 7, 1)$. Since $\fa=6$, then 
by Equations (\ref{eq:row-one}) and (\ref{eq:row-last}), we have
$$\alpha_0\alpha_1 = (9, 11,  6,  4, 12,  6, 10) \text{ and } 
  \alpha_0\alpha_{n-1} = (10, 6, 4, 12, 6, 10, 10).$$
The three remaining rows i.e., $\alpha_0^2, \ \alpha_0\alpha_1,$ and 
$\alpha_0\alpha_{n-1}$ are full of non-zero elements, then $\cN_1$ 
achieves the upper bound $43$. We note that the sum defined by the 
$4$ rows at the middle of the multiplication table is equal to $22$.

\subsection{Second example}\label{sub:SE}
In this subsection, we construct two elliptic normal bases of $\FF_{7^6}$ over 
$\FF_{7}$ using the elliptic curves $$E_2 : y^2+3xy+2y = x^3+x^2+2x+4$$ and
$$E_3 : y^2+3xy+4y = x^3+6x^2+1$$ of order $12$ over $\FF_7$.  
The aim is to show how the equation of the chosen curve can affect the 
complexity of the derived basis.

\subsubsection{Analysis of the elliptic normal basis from $E_2$}
The point $t_2=(2, 2)$ generates a subgroup $\sT_2\subset E_2$ of order $6$. 
With the quotient isogeny $E_2'=E_2/\sT_2$ defined by $$y^2+3xy+2y=x^3+x^2+2x+4,$$
we find the point $a_2=(4, 3)$ with irreducible preimage that generates the
finite field $\FF_{7^6}$. We set $$b_2=(5\tau^5+5\tau^4+2\tau^3+2\tau^2+\tau+4,
6\tau^5+\tau^4+2\tau^3+2\tau^2+3\tau+1)$$ and we find the elliptic normal basis
$\cN_2$ given by $\alpha=2\tau^4+2\tau^3+\tau$ and its conjugates. With the
point $R=(4, 1)$, lying outside of $E_2[6](\FF_7)$, we establish 
Table~\ref{table:specvects2} that contains the special vectors of $\cN_2$.

\begin{table}[ht]
\centering 

\begin{tabular}{cc||c} 
\hline 
& & 
($\overrightarrow{\cR}_2^{-1}*\overleftarrow{\cR}_{2_k})_{2\leq k\leq n-2}$ \\ [0.5ex] 
\hline
$\overrightarrow{\cR}_{2x}$ & (4, 4, 3, 1, 1, 3) &  (5, 6, 5, 4, 4, 4) \\ 
$\overrightarrow{\cR}_2$ & (4, 2, 4, 0, 5, 0) &  (4, 3, 3, 4, 3, 3) \\
$\overrightarrow{\cR}_2^{-1}$ & (1, 0, 3, 0, 1, 3) & (5, 4, 4, 4, 5, 6)  \\
\hline 
\end{tabular}
\caption{Special vectors of $\cN_2$.}
\label{table:specvects2} 
\end{table}

Table \ref{table:specvects2} gives an example where the special vectors, 
that bound the complexity, are full of nonzero elements. Thus, without 
surprise, this leads to higher bounds on the complexity of $\cN_2$
$$21\leq C(\cN_2)\leq31.$$
Using our function \textit{ENB3RowsWeight}, we found that 
$\alpha_0^2=(4, 2, 6, 0, 2, 0)$ and $\overrightarrow{\iota}=(5, 1, 5, 1, 0, 2)$. 
Since $\fa=2$, using Equations (\ref{eq:row-one}) and (\ref{eq:row-last}), 
we have
$$\alpha_0\alpha_1 = (3, 6, 4, 0, 2, 1) \text{ and } 
  \alpha_0\alpha_{n-1} = (5, 0, 4, 0, 2, 0).$$
The number of non-zero elements in the three remaining rows is equal to 
$12$. Thus, the total complexity of $\cN_2$ is 30. In the next section, 
using different parameters, we give an elliptic normal basis that has a 
better lower bound and a better complexity.

\subsubsection{Analysis of the elliptic normal basis from $E_3$}\label{subsubsec:AENBE3}

In this example, we deal with the elliptic curve $E_3$. The torsion point of
order $6$, the isogenous curve $E_3'$, the point $a_3$ with irreducible preimage,
and the generic point $b_3$ are respectively given by
\begin{align*}
t_3  & =  (4, 5), \\
E_3' & :  y^2+3xy+4y = x^3+6x^2+4x+5, \\ 
a_3  & =  (0, 1), \\
b_3  & =  (4\tau^5+3\tau^4+2\tau^3+5\tau^2+3\tau+6, 5\tau^4+6\tau^3+\tau^2+2\tau+5).
\end{align*}
The above parameters give an elliptic normal basis $\cN_3$ generated by the normal element
$$\alpha_3 = 3\tau^5+3\tau^4+3\tau^3+3\tau^2+1.$$
Using the point $R_3=(1, 1)$, that lies outside of $E_3[6](\FF_7)$, we find the
special vectors in Table~\ref{table:specvects3}.

\begin{table}[ht]
\centering 
\begin{tabular}{cc||c} 
\hline 
& & 
($\overrightarrow{\cR}_3^{-1}*\overleftarrow{\cR}_{3_k})_{2\leq k\leq n-2}$ \\ [0.5ex] 
\hline
$\overrightarrow{\cR}_{3x}$ & (1, 1, 6, 2, 2, 6) & (3, 4, 3, 0, 0, 0) \\ 
$\overrightarrow{\cR}_3$ & (2, 5, 2, 1, 4, 1) & (3, 0, 0, 3, 0, 0) \\
$\overrightarrow{\cR}_3^{-1}$ & (0, 3, 1, 3, 0, 1) & (3, 0, 0, 0, 3, 4) \\ 
\hline 
\end{tabular}
\caption{Special vectors of $\cN_3$.}
\label{table:specvects3} 
\end{table}

We observe that the weights of the special vectors 
$(\overrightarrow{\cR}_3^{-1}*\overleftarrow{\cR}_{3_k})_{2\leq k\leq n-2}$
in Table \ref{table:specvects3} are sparser than those reported in 
Tables \ref{table:specvects1} and \ref{table:specvects2}. This leads 
to a better complexity bound for the elliptic normal basis $\cN_3$
$$11\leq C(\cN_3)\leq26.$$

For this example, $\alpha_0^2=(4, 6, 0, 1, 2, 0)$ and the precomputed vector
$\overrightarrow\iota$ that corresponds to the $x-$coordinate of $b_3$ is 
equal to $(1, 3, 1, 0, 0, 1)$. Then, by Equations (\ref{eq:row-one}) and 
(\ref{eq:row-last}), we have 
$$\alpha_0\alpha_1 = (6, 0, 2, 1, 5, 0) \text{ and } 
  \alpha_0\alpha_{n-1} = (6, 2, 1, 5, 0, 0),$$
and the complexity of $\cN_3$ is $20$. We comment that the sum depending 
on $R_3$ and $t_3$ is equal to $8$. This example shows that even if the 
first two rows $\alpha_0^2$ and $\alpha_0\alpha_1$, and the last row 
$\alpha_0\alpha_{n-1}$ have many nonzero elements, we can end up with a 
better bound and a better complexity if the sum defined by the rows in 
the middle of the multiplication table is small.


\section{Conclusions and further work} \label{sec:KNRP}

Looking at our results in Section \ref{sec:CEP}, 
we observe that it will be very hard to get a low complexity normal 
basis, as designed in \cite{CL09}, if we can not ensure low weight 
vectors from inverses and convolution products between some special 
vectors. This likely involves that vectors with prescribed low 
weight are needed in the computations. To have a concrete view of 
this analysis let us consider the matrices 
$$M_2=
\begin{pmatrix}
0 & 1 & 1 & 0 & 0 & 0\\
6 & 0 & 2 & 0 & 6 & 0\\
0 & 3 & 0 & 0 & 3 & 0\\
2 & 0 & 6 & 0 & 6 & 0\\
1 & 1 & 0 & 0 & 0 & 0\\
2 & 4 & 2 & 0 & 4 & 0
\end{pmatrix}
$$
and
$$M_3=
\begin{pmatrix}
2 & 3 & 3 & 2 & 4 & 4\\
1 & 5 & 4 & 5 & 1 & 1\\
2 & 6 & 2 & 2 & 6 & 2\\
4 & 5 & 1 & 1 & 1 & 5\\
3 & 3 & 2 & 4 & 4 & 2\\
4 & 4 & 4 & 1 & 2 & 1\\
\end{pmatrix}
.$$ 
The rows of $M_2$ are equal to the vectors 
$(\overleftarrow{\cR}_{2_k})_{1\leq k\leq n}$, the component-wise 
product between $\overrightarrow{\cR}_2$ and its $k$-cyclic shift; 
the matrix $M_3$ is constructed similarly. These products play a 
central role in the proof of Theorem~\ref{The:Complexity:EllPeriod} 
and the construction of the normal bases $\cN_2$ and $\cN_3$ in 
Section \ref{sub:SE}. We can apply the same construction on the 
normal basis $\cN_1$, of Section \ref{sub:FE}, but the cases $\cN_2$ 
and $\cN_3$ are sufficient for the current analysis. All the given 
data are obtained from the Magma package and the practical examples
are available in \cite{MagmaPack}.

The matrices $M_2$ and $M_3$ give rise, respectively, to the bounds of the 
elliptic normal bases $\cN_2$ and $\cN_3$. Even though, $M_3$ is full of 
nonzero elements and $M_2$ has a better complexity, we end up with a good 
bound for $\cN_3$, which is not the case for $\cN_2$. This example shows 
clearly that it is hard to predict the weight of the special vectors 
$(\overrightarrow{\cR}^{-1}*\overleftarrow{\cR}_{k})_{2\leq k\leq n-2}$
from its entries. Since the complexity of elliptic normal basis depends 
on these products, it is a nontrivial and interesting research 
problem to obtain some points on elliptic curve that lead to vectors 
with low weight. We leave this for further research.

\vspace{0.5cm}

\noindent {\bf Acknowledgements.} The authors were funded by the Natural
Sciences and Engineering Research Council of Canada.

\end{document}